\documentclass[letterpaper, 10 pt, conference]{ieeeconf}
\usepackage{amsmath}
\usepackage{amsfonts,amssymb}
\usepackage{tikz} 
\usepackage{xr}
\usepackage{algorithm}
\usepackage{algorithmicx}
\usepackage{algpseudocode}
\usepackage{cite}
\usepackage{bbm}

\newtheorem{theorem}{Theorem}

\newtheorem{corollary}[theorem]{Corollary}
\newtheorem{remark}[theorem]{Remark}
\newtheorem{definition}[theorem]{Definition}

\IEEEoverridecommandlockouts                              
\title{\LARGE \bf
A Moving Target Approach for Identifying Malicious Sensors in Control Systems
}
\author{Sean Weerakkody~~~~Bruno Sinopoli
\thanks{S Weerakkody, and B. Sinopoli are with the Department of Electrical and Computer Engineering, Carnegie Mellon University, Pittsburgh, PA, USA 15213. Email: {\tt\small sweerakk@andrew.cmu.edu, brunos@ece.cmu.edu.}}
\thanks{S. Weerakkody is supported in part by the Department of Defense (DoD) through the National Defense Science \& Engineering Graduate Fellowship (NDSEG) Program. The work by S. Weerakkody and B. Sinopoli is supported in part by the Department of Energy under Award Number DE-OE0000779 and by the National Science Foundation under Award Number 1646526”.}}

\begin{document} \maketitle
\begin{abstract}
In this paper, we consider the problem of attack identification in cyber-physical systems (CPS). Attack identification is often critical for the recovery and performance of a CPS that is targeted by malicious entities, allowing defenders to construct algorithms which bypass harmful nodes.  Previous work has characterized limitations in the perfect identification of adversarial attacks on deterministic LTI systems. For instance, a system must remain observable after removing any $2q$ sensors to only identify $q$ attacks. However, the ability for an attacker to create an unidentifiable attack requires knowledge of the system model. In this paper, we aim to limit the adversary's knowledge of the system model with the goal of accurately identifying all sensor attacks. Such a scheme will allow systems to withstand larger attacks or system operators to allocate fewer sensing devices to a control system while maintaining security. We explore how changing the dynamics of the system as a function of time allows us to actively identify malicious/faulty sensors in a control system. We discuss the design of time varying system matrices to meet this goal and evaluate performance in deterministic and stochastic systems.
\end{abstract}
\section{Introduction}
Cyber-physical system (CPS) security has become a widely studied topic, both within industry and the research community. Indeed, current and next generation cyber-physical systems will pervade our critical infrastructures including the electricity grid, water distribution systems, waste management, smart buildings, and health care, providing ample motivation for attackers to target CPS. Moreover, CPS will consist of many heterogeneous systems and components, thus providing adversaries an opportunity for attacks \cite{Cardenas:2008ke}. There have been several powerful attacks against CPS including Stuxnet \cite{Langner2013} and the Maroochy Shire incident \cite{Slay2008}.

Consequently, significant research effort has been geared towards responding to attacks on CPS, with a focus on integrity attacks where an adversary alters a subset of control inputs and sensor measurements. Previous research has characterized stealthy attack scenarios including zero dynamic attacks \cite{teixeira2012, PasqualettiJournal}, false data injection  attacks \cite{liu2009,moscs10security}, covert attacks \cite{Smith2011}, and replay attacks \cite{Mo2009R}, where an adversary adversely affects a system without being detected. Moreover, countermeasures have been developed to detect stealthy attacks \cite{teixeira2012pt2, Mo2014, Weerakkody2014, Miao2014, yuan2015security}.

While detection allows a defender to determine the presence of an adversary, it does not prescribe a solution to maintain system performance. Thus, we aim to identify malicious entities in the network so specialized and optimal solutions can be developed to counter their presence. 

Previous work has considered the problem of identification in CPS. For instance, Pasqualetti et al. \cite{PasqualettiJournal} define the notion attack identifiability, provide algebraic conditions to characterize when a system can identify $q$ attacks on sensors and actuators, and propose identification filters to carry out this process. Also, Sundaram et al. \cite{sundaram2010wireless} design an intrusion detection system in a wireless control network to identify malicious behavior at sensor nodes and provide graphical conditions to determine when identification is possible.

Attack identification is closely tied to robust estimation in the presence of sensor attacks \cite{fawzi2011secure, fawzi2014secure, chong2015observability, shoukry2013event, mishra2015secure, nakahira2015dynamic, pajic2014robustness}. Qualitatively, by identifying sensor attacks, a defender can perform state estimation by disregarding malicious sensors and carrying out estimation schemes on trusted sensors. In fact, in deterministic systems, the ability to perform perfect state estimation in the presence of sensor attacks is equivalent to being able to perform perfect identification \cite{fawzi2011secure}.

The problem of robust state estimation in deterministic systems is considered in \cite{fawzi2011secure, fawzi2014secure, shoukry2013event, chong2015observability}. For instance, in \cite{fawzi2011secure} (and \cite{fawzi2014secure}) the authors characterize the number of attacked sensors (and inputs) which can be tolerated while performing state estimation and propose a decoding algorithm to recover the state. Next, Shoukry et al. \cite{shoukry2013event} propose event based algorithms to improve the efficiency of robust state estimation. Additionally, Chong et al. \cite{chong2015observability} develop schemes for robust estimation in deterministic continuous LTI systems while formulating a notion of observability under attack.

The problem of robust estimation has also been considered in stochastic and uncertain systems \cite{mishra2015secure, nakahira2015dynamic, pajic2014robustness}. For instance, Mishra et al. \cite{mishra2015secure} propose robust estimation schemes in the presence of Gaussian noise and characterize limitations in estimation performance in the presence of attacks. Nakahira et al. \cite{nakahira2015dynamic} consider robust estimation with bounded noise, proposing a stable estimator in the presence of $q$ attacks provided that the system remains detectable after removing any $2q$ sensors. Finally, Pajic et al. \cite{pajic2014robustness} demonstrate the robustness of estimation schemes in the presence of attacks even when there is uncertainty in the system model.

From previous work \cite{PasqualettiJournal,fawzi2014secure,chong2015observability, shoukry2013event} it can be shown that there are fundamental limitations in the defender's ability to perform identification. For instance, given the presence of an intelligent adversary, perfect identification of $q$ attacks can only be performed if the system is observable after removing any $2q$ sensors. However, such intelligent adversaries require knowledge of the system model. By limiting the attacker's knowledge of the model, we hope to identify all sensor attacks. This will allow us to perform perfect estimation in a deterministic system under $q$ attacks if the system is observable after removing any $q$ sensors. To meet this objective, we propose changing the dynamics of the system model in a time varying fashion, unknown to the adversary. A similar scheme is considered for the detection of integrity attacks in  \cite{weerakkody2015detecting}. The time varying dynamics act like a moving target, preventing an information constrained adversary from developing stealthy unidentifiable attacks. In this paper, 
\begin{itemize}
\item We propose design considerations for the moving target, which allow us to perfectly identify a class of nonzero sensor attacks in deterministic systems. 
\item We show the moving target allows us to detect destabilizing sensor attacks in stochastic systems.
\item We construct an observer which allows us to identify sensor attacks that cause unbounded estimation error. 
\end{itemize}

The rest of the paper is summarized as follows. In section \ref{problem setup}, we introduce the problem and results on identification. In section \ref{Moving Target}, we propose a moving target based formulation to prevent unidentifiable attacks. Next, in section \ref{Design}, we consider deterministic identification and propose design considerations for the moving target. After, in section \ref{Stochastic}, we consider the moving target approach for detecting false data injection attacks. In section \ref{Stochastic2}, we construct a robust estimator which can only be destabilized by identifiable sensor attacks. 
 Section \ref{Conclusion} concludes the paper. \\
\textit{Notation:} Unless otherwise specified, the following notation is used. $x_{t_1:t_2}$ refers to the set $\{ {x}_{t_1}, x_{t_1+1}, \cdots, x_{t_2} \}$. $X_{uv}$ is the entry at row $u$ and column $v$ of matrix $X$. $X^i$ is the $i$th row of $X$. If $V$ is a set of indices, $X^V$ are the rows of $X$ indexed by $V$. $\mathbf{1}$ is the indicator function. $\mathbb{P}(\cdot)$ refers to the power set. $\|\cdot\|$ is the $l_2$ norm. $\{a_k\}$ defines a sequence.

\section{Problem Setup and Previous Results} \label{problem setup}
In this section we introduce our system model. To begin we model our system as a discrete time deterministic control system under sensor attacks. Though this initial approach is simple, we consider the stochastic case later in the paper. The dynamics are given by
\begin{align}
x_{k+1} &= Ax_k + B(u_k(y_{0:k})),~~y_k = Cx_k + D d_k^a. 
\end{align}
where $x_k \in \mathbb{R}^n$ is the state at time $k$, $u_k(y_{0:k}) \in \mathbb{R}^p$ is the control input, and $y_k \in \mathbb{R}^m$ are the sensor outputs. The sensor outputs, $y_k$, consist of $m$ scalar sensor outputs, defined by the set $S = \{1,2,\cdots,m\}$. 

The adversary performs sensor integrity attack on an ordered set $K = \{s_1, s_2, \cdots, s_{|K|}\} \subseteq S$  using additive inputs $d_k^a \in \mathbb{R}^{|K|}$, starting at time $k = 0$. Consequently, we define $D$ entrywise as
\begin{equation}
D_{uv} = \mathbf{1}_{u = s_i, v = i}.
\end{equation}
Note that $D$ is fully determined by the set $K$. Implicitly, we assume that the set of sensors which the adversary targets is constant due to (ideally) the inherent difficulty in the task of hijacking sensors. When performing an integrity attack, the adversary's goal is to adversely affect the physical system by preventing proper feedback. In particular, a defender with incorrect sensor measurements may not be able to perform adequate state estimation and thus will not be able to apply appropriate corrective measures to the system. 

We assume that the defender knows the system dynamics $\mathcal{M} = \{A,B,C\}$ as well as the input and output histories given by $u_{0:k}$ and $y_{0:k}$, but is unaware of the set $K$. Furthermore, we assume that in the deterministic setup, the defender is unaware of the initial state $x_0$. Thus, the defender's information $\mathcal{I}_k$ at time $k$ is given by
\begin{equation}
\mathcal{I}_k = \{\mathcal{M},u_{0:k-1},y_{0:k}\}.
\end{equation}
\begin{remark}
In the deterministic case, we explore attacks where the defender has no knowledge of the initial state. While this is certainly not realistic, the attack vectors developed in this scenario can still remain stealthy in a practical stochastic setting if the adversary carefully ensures that his initial attack inputs remain hidden by the noise of the system.
\end{remark}

From a defender's perspective it is important to identify trusted sensor nodes. Estimation and control algorithms can then be tuned to ignore attacked nodes. We note that the problem of identifying malicious nodes is independent of the control input, since the defender is aware of the model and input history. Thus, in the ensuing discussions we will disregard the control input so that
\begin{align}
x_{k+1} &= Ax_k, ~~y_k = Cx_k + D d_k^a \label{eq:output}. 
\end{align}
\subsection{Previous Results in Identification}
Previous work in CPS security has attempted to address the problem of identifiability. We now revisit major results when trying to identify malicious adversaries. 

To begin, let $y(x_0,D(K) d_k^a, k)$ be the output signal associated with \eqref{eq:output} at time $k$ with initial state $x_0$ and attack input $D d_k^a$ on sensors $K$. We define an attack monitor $\Psi: y_{0:\infty} \rightarrow \mathbb{P}(\{1,\cdots,m\})$ with the property that $\Psi = K$ if and only if $K$ is the smallest set such that for all $k \ge 0$ we have $y_k = y(x_0^*,D(K) d_k^a,k)$ for some initial state $x_0^*$.
\begin{definition}
Consider an attack $\{D(K) d_k^a\}$ on the sensors in $K$, where we assume each sensor in $K$ is attacked at least once. Then $\{D(K)d_k^a\}$ is unidentifiable if $\Psi( y_{0:\infty}) \neq K$. Moreover, the set $K$ is unidentifiable if there exists an unidentifiable attack \cite{PasqualettiJournal}.
\end{definition}
Based on this definition we state an attack is unidentifiable if there exists an attack on a different subset of sensors with size less than or equal to $|K|$, which yields the same sensor outputs. As a result, we have the following \cite{PasqualettiJournal}.
\begin{theorem}
An attack on sensors $K_1$ is unidentifiable if and only if there exists $x_0^1$ and $x_0^2$ and set $K_2$ such that \label{attackinput}
\begin{equation}
y(x_0^1,D(K_1)d_k^a,k) = y(x_0^2,D(K_2) \bar{d}_k^a,k), ~~~ \forall ~k.
\end{equation}
where $|K_2| \le |K_1|$ and $K_2 \neq K_1$.
\end{theorem}
In this setup, it is important to include the condition that $|K_2| \le |K_1|$. Without this restriction, an attacker can generate an unidentifiable attack regardless of the set $K$. In particular, suppose
$D(K)d_k^a = C^K A^k \Delta x_0$. Then, 
\begin{equation}
y(x_0, D(K)d_k^a ,k) = y(x_0 + \Delta x_0 ,D(K^c) \bar{d}_k^a,k), ~~~ \forall ~k.
\end{equation}
where $D(K^c) \bar{d}_k^a = -C^{K^c} A^k \Delta x_0$. We observe that there exists an attack on sensors $K$ equivalent to an alternative attack on the complement set $K^c$. In such a scenario, it is easy to see a defender can identify attacks on no more than half the sensors. In fact, we have the following result \cite{PasqualettiJournal}.
\begin{theorem}
Attack set $K_1$ is unidentifiable if and only if 
\begin{equation}
\begin{bmatrix} \lambda I - A & 0 & 0 \\ C & D(K_1) & D(K_2) \end{bmatrix} \begin{bmatrix} x \\ d_1 \\ d_2 \end{bmatrix} = \begin{bmatrix} 0 \\ 0 \end{bmatrix},
\end{equation}
where $K_2$ satisfies $|K_2| \le |K_1|$ and $K_2 \neq K_1$, $\lambda \in \mathbb{C}$, $x \neq 0 \in \mathbb{C}^n$ and $d_1 \in \mathbb{C}^{|K_1|}$ and $d_2 \in \mathbb{C}^{|K_2|}$.
\end{theorem}
\begin{corollary} \cite{fawzi2011secure}
A system can identify attacks on up to $q$ arbitrary sensors if and only if $(A,C)$ is $2q$ sparse observable. That is, $(A,C)$ remains observable if we remove any $2q$ sensors. \label{cor}
\end{corollary}
As a consequence of the previous results, the defender must allocate a large number of sensors to withstand a large number of attacks or alternatively design systems, which are only robust to a fewer attacks.

\section{Moving Target Approach for Identification} \label{Moving Target}
In the previous section we demonstrated that there exist limitations on the number of attacks a defender can potentially identify. Specifically, we saw that to identify all attacks of size $q$, the system must be $2q$ sparse observable. This can result in expenditures to add more sensing in order to withstand more attacks or sacrificing security in order to use fewer components. However, in this section we argue that generating unidentifiable attacks requires knowledge of the model. By limiting this knowledge, we hope to prevent such attacks. To begin we define the following.
\begin{definition}
A nonzero attack on sensor $s$ is unambiguously identifiable at time $t$ if there is no $x_0^* \in \mathbb{R}^n$ satisfying $y_k^s = y^s(x_0^*, 0, k)$ for $0 \le k \le t$. An attack on sensor $s$ is unambiguously identifiable  if it is unambiguously identifiable for all $t$.
\end{definition}
The notion of unambiguous identifiability characterizes when the defender can be certain that sensor $s$ is faulty or under attack. This scenario occurs only if there exists no initial state which produces the output sequence at $y^s$. We envision designing a system that forces the attacker to generate unambiguously identifiable attacks on all sensors which he targets. Consequently, we can identify misbehaving sensors. 

Thus, instead of requiring our system to be $2q$ sparse observable to perform perfect estimation with $q$ attacks or $2q$ detectable to perform stable estimation \cite{nakahira2015dynamic}, forcing an attacker to generate unambiguously identifiable attacks will allow the defender to perform stable estimation when the system is only $q$ detectable (detectable after removing any $q$ sensors). This allows the system to withstand more powerful attacks or use fewer sensing devices while maintaining the same level of security. We now characterize attacks which are not unambiguously identifiable.
\begin{theorem} \label{obsvmatinput}
An attack on sensor $s$ is not unambiguously identifiable at time $t$ if and only if there exists an $x_0^*$ such that $D^s d_k^a = C^s A^k x_0^*$ for all  $0 \le k \le t$ and $ C^s A^k x_0^* \neq 0$ for some time $0 \le k \le t$. 
\end{theorem}
\begin{proof}
Suppose $D^s d_k^a = C^s A^k x_0^*$ for time $0 \le k \le t$. Assume this attack is nonzero. Then, $y_k^s = y^s(x_0+x_0^*,0,k)$. Suppose instead that there is no $x_0^*$ such that $D^s d_k^a \neq C^s A^k x_0^*$ for $0 \le k \le t$ . Then there is no $\bar{x}_0$ such that $y_k^i+D^sd_k^a$ = $C^sA ^k \bar{x}_0$.  Since $y^s(x_0,0,k) = C^s A^k x_0$, the result immediately follows.
\end{proof}
As a result, to prevent attacks on sensor $i$ from being unambiguously identifiable at time $k$, an adversary must insert attacks which lie in the image of $\mathcal{O}_{k+1}^s$ given by
\begin{equation}
\mathcal{O}_{k+1}^s = \begin{bmatrix} {C^s}^T  (C^sA)^T \cdots (C^sA^k)^T \end{bmatrix}^T.
\end{equation}
To insert such attacks, the adversary likely has to be aware of both the matrix $A$ and the matrix $C^s$. In the sequel, we aim to minimize this knowledge to prevent an attacker from generating unidentifiable attacks.
\subsection{A moving target approach}
Ideally, we would like to simply assume the adversary has no knowledge of $(A,C)$ and consequently will likely always be unambiguously identifiable. However, in practice, the processes associated with the physical plant may be well known or previously public so that the attacker is aware of $(A,C)$. Alternatively, the defender can change parameters of the system to ensure a knowledgeable adversary is still thwarted.  Specifically, we propose changing the system matrix $A$ and $C$ in a time varying and unpredictable fashion from the adversary's point of view so that
\begin{align}
x_{k+1} &= A_k x_k, ~~~ y_k = C_k x_k + D d_k^a. \label{eq:timevarying}
\end{align}
We assume that 
\begin{displaymath}
(A_k,C_k) \in \Gamma = \{(A(1),C(1)),\cdots,(A(l),C(l))\}.
\end{displaymath}
\begin{theorem}
An attack on sensor $s$ in \eqref{eq:timevarying}  is not unambiguously identifiable at time $t$ if and only if there exists an $x_0^*$ such that $D^s d_k^a = C_k^s (\prod_{j=0}^{k-1} A_j)x_0^*$ for all time $0 \le k \le t$ and $  C_k^s (\prod_{j=0}^{k-1} A_j)x_0^* \neq 0$ for some time $0 \le k \le t$. \label{MTImage}
\end{theorem}
\begin{proof}
The proof is similar to that of Theorem \ref{obsvmatinput}.
\end{proof}
Changing the system matrices as a function of time allows the system to act like a moving target. In particular, even if an attacker is aware of the existing configurations of the system, defined by $\Gamma$, he will likely be forced to generate unambiguously identifiable attacks since he is not aware of the sequence of system matrices. Moreover, since the system matrices keep changing, it is unlikely the attacker can remain unidentifiable by pure chance.
\begin{remark}
The matrices $(A_k,C_k)$ can be changed randomly using a cryptographically secure pseudo random number generator where the random seed is known both by the defender and the plant, but is unavailable to the adversary. From a security perspective, the seed would form the root of trust. The set $\Gamma$ can be obtained by leveraging or introducing degrees of freedom in the dynamics and sensing in our control system. While the defender likely would have to change his control strategy to account for the time varying dynamics, in this work, we will ignore such changes.
\end{remark}
Given the proposed setup, we are now ready to define the attacker's information and an admissible attacker strategy.
\textbf{Attacker Information}
\begin{enumerate}
\item  The adversary has no knowledge of either the input sequence $u_{0:k}$ or the true output sequence.
\item  $\mathcal{I}_k^a = \{\Gamma, Dd_{0:k-1}^a, f(\{(A_k,C_k)\})\}$.
\end{enumerate}
If the adversary can observe the output sequence in a zero input deterministic setting, he can multiply the true outputs by some constant factor to avoid generating  unambiguously identifiable inputs. A realistic adversary may use physical attacks to bias sensors without reading their outputs. Future work will examine relaxing this assumption. The control inputs are also secret so that the attacker will be unable to leverage the input process to gain information about the system model. However, we assume $\Gamma$ is known as well as the sequence of attack inputs. Also, the probability distribution of the sequence of system matrices, $f(\{(A_k,C_k)\})$, is public.
\begin{definition}
An admissible attack policy is a sequence of deterministic mappings $\Omega_k: \mathcal{I}_k^a \rightarrow \mbox{Im}(D(K))$ such that $Dd_k^a = \Omega_k(\mathcal{I}_k^a)$.
\end{definition}
Here, we assume the attacker can only leverage his information to construct a stealthy attack input. Consequently, while there may exist attacks that bypass identification, in order to be admissible, they must leverage the attacker's knowledge and can not be a function of unknown and unobserved stochastic processes (namely the sequence of $\{A_k\}$ and $\{C_k\}$). A real adversarial strategy may be to bias sensors with the goal of affecting state estimation, without being identified by the defender. Thus, the adversary can impact the system without corrective measures being put in place.

\section{System Design for Deterministic Identification} \label{Design}
In this section, we consider criteria the defender can use to design the set $\Gamma$, which can allow him to identify malicious inputs on a subset of sensors. 
Given the attacker's knowledge of $\Gamma$, an adversary can guess the sequence of system matrices chosen by the defender. If the adversary guesses correctly, he can generate attacks which are not unambiguously identifiable. We would now like to characterize the scenario where an attacker can guess the sequence of matrices incorrectly yet still generate an unambiguously identifiable attack. 
\begin{theorem} \label{sum of nulls}
Suppose an adversary generates an attack on sensor $s$ by guessing a sequence $\{l_k\}$ where $l_i \in \{1, \cdots, l\}$ and creating inputs by applying Theorem \ref{MTImage}. Specifically, there exists an $x_0^1$ such that $D^s d_k^a = C^s(l_k) (\prod_{j=0}^{k-1} A(l_j))x_0^1$ for all time $0 \le k \le t$ and $ D^s d_{\eta}^a \neq 0$ for some time $0 \le \eta \le t$. Such a strategy may avoid generating an unambiguously identifiable attack on sensor $s$ at time $t$ if and only if
\begin{equation}
 \mbox{null} \begin{pmatrix}  \mathcal{O}(l_s,t) & \mathcal{O}(s,t)  \end{pmatrix} > \mbox{null}  \begin{pmatrix} 
  \mathcal{O}(l_s,t) \end{pmatrix} + \mbox{null} \begin{pmatrix}  \mathcal{O}(s,t) \end{pmatrix},  \label{eq:attackcond}
  \end{equation}
\begin{align*}
 \mathcal{O}(l_s,t) &= \begin{bmatrix} (C^s(l_0))^T  & \cdots & (C^i(l_t)\prod_{j=0}^{t-1}A(l_j))^T \end{bmatrix}^T, \\
 \mathcal{O}(i,t)  &=  \begin{bmatrix} (C_0^s)^T &  (C_1^sA_0)^T &  \cdots & (C_t^s \prod_{j=0}^{t-1}A_j)^T \end{bmatrix}^T,
\end{align*}
where null refers to the dimension of the null space.
\end{theorem}
\begin{proof}
From Theorem \ref{MTImage}, an attack is not unambiguously identifiable at time $t$ if and only if there exists some $x_0^2$ such that  $D^s d_k^a = C_k^s (\prod_{j=0}^{k-1} A_j) x_0^2$ for $0 \le k \le t$ and this sequence is nonzero. Thus, the proposed strategy can generate a nonzero unambiguously identifiable attack on sensor $i$ at time $t$ if and only if
\begin{equation*}
C^s(l_k)\left(\prod_{j=0}^{k-1}A(l_j)\right) x_0^1 =  C_k^s \left(\prod_{j=0}^{k-1}A_j\right) x_0^2,
\end{equation*}
for all $0 \le k \le t$ and moreover for some $0 \le k \le t$ this expression is nonzero. The result immediately follows.
\end{proof}

In practice, it is unlikely that the defender can change the parameters of the system at each time step due to the system's inertia. Consequently, we would like to consider systems where $(A_k,C_{k})$ remains constant for longer periods of time. For now, we assume $(A_{k},C_{k}) \subset \{\Gamma\}$, but is \textit{constant}. An adversary, can use his knowledge of $\Gamma$ to guess a pair  $(A_{k},C_{k}) \in \Gamma$ and generate unidentifiable attack inputs. Define the matrix
\begin{equation}
\mathcal{O}_{t,j}^{\mathcal{S}} = \begin{bmatrix} {C^{\mathcal{S}}(j)}^T  (C^{\mathcal{S}}(j)A(j))^T \cdots (C^{\mathcal{S}}(j)A(j)^{t-1})^T \end{bmatrix}^T.
\end{equation}

If the attacker guesses the matrices $(A(j),C(j))$ and chooses to attack sensor ${s}$, he would need to ensure $\begin{bmatrix} (D^s d_0^a)^T  & \cdots &(D^s d_t^a)^T \end{bmatrix}^T$ lies in the image of $\mathcal{O}_{t+1,j}^{{s}}$ to avoid deterministic identification. We next determine when an attacker is able to guess an incorrect pair and  avoid generating an unambiguously identifiable attack. 

\begin{theorem}
 Suppose $(A,C) = (A(1),C(1))$ and an adversary generates a nonzero attack input on sensor $s$ using $(A(2),C(2))$ by inserting attacks along the image of $\mathcal{O}_{t,2}^{{s}}$. Let $\Lambda^1 = \{\lambda_1^1, \cdots, \lambda_{q_1}^1\}$ be the set of distinct eivenvalues associated with $A(1)$ and $\Lambda^2 = \{\lambda_1^2, \cdots, \lambda_{q_2}^2\}$ be the set of distinct eigenvalues of $A(2)$. Let 
 \begin{displaymath}
 \{v_{1,1}^{\lambda,j},\cdots v_{r_1,1}^{\lambda,j}, v_{1,2}^{\lambda,j},\cdots v_{r_2,2}^{\lambda,j} , \cdots, v_{1,l_{\lambda,j}}^{\lambda,j},\cdots v_{r_{l_{\lambda,j}},l_{\lambda,j}}^{i,j} \}
\end{displaymath}
be a maximal set of linearly independent (generalized) eigenvectors associated with eigenvalue $\lambda$ of $A(j)$ satisfying
\begin{equation}
A(j) v_{1,l}^{\lambda,j} = \lambda v_{1,l}^{\lambda,j}, ~~~~ A(j)v_{k+1,l}^{\lambda,j} = \lambda v_{k+1,l}^{\lambda,j} + v_{k,l}^{\lambda,j}.
\end{equation}
Noting that each $r_i$ is in general fully determined by $\lambda$ and $j$, let $r(\lambda) = \max_{i,j} r_i(\lambda,j)$.
Define $V_{s,k}^{\lambda,j} \in \mathbb{C}^{r(\lambda) \times r_k}$ as 
\begin{equation*}
 \begin{bmatrix} C^s(j) v_{1,k}^{\lambda,j} & C^s(j) v_{2,k}^{\lambda,j} &  \cdots & \cdots & C^s(j) v_{r_k,k}^{\lambda,j} \\											\mathbf{0} & C^s(j) v_{1,k}^{\lambda,j} & \ddots  & \ddots & \vdots \\
 									\mathbf{0} &  \mathbf{0} &  \ddots   & \ddots & \vdots \\
 									\mathbf{0} &  \mathbf{0} & \cdots  &   C^s(j) v_{1,k}^{\lambda,j} & C^s(j) v_{2,k}^{\lambda,j} \\
 									\mathbf{0} &  \mathbf{0} & \mathbf{0}  & \mathbf{0} &  C^s(j) v_{1,k}^{\lambda,j} \\
 									\mathbf{0} &  \mathbf{0} & \mathbf{0}  & \mathbf{0} & \mathbf{0} \end{bmatrix}.
\end{equation*}
There exists an attack on sensor $s$, which is not unambiguously identifiable for all time if and only if $\Lambda^1 \cap \Lambda^2 \neq \emptyset$ and there exist some $\lambda \in \Lambda^1 \cap \Lambda^2$ such that 
\begin{equation*}
\mbox{null} \begin{pmatrix} \mathcal{V}_s^{\lambda,1} & \mathcal{V}_s^{\lambda,2} \end{pmatrix} > \mbox{null} \begin{pmatrix} \mathcal{V}_s^{\lambda,1} \end{pmatrix} + \mbox{null} \begin{pmatrix} \mathcal{V}_s^{\lambda,2} \end{pmatrix},
\end{equation*}
where
\begin{equation*}
\mathcal{V}_s^{\lambda,j} = \begin{pmatrix}  V_{s,1}^{\lambda,j} &  \cdots &  V_{s,l_{\lambda,j}}^{\lambda,j} \end{pmatrix}.
\end{equation*}
Otherwise the attack can be detected in time $t \le 2n-1$. \label{theorem:det}
\end{theorem}
\begin{proof}
The proof is lengthy and found in \cite{weerakkody2016information} along with numerical simulations.
\end{proof}
The previous theorem gives the defender an efficient way to determine if the attacker can guess $\Gamma$ incorrectly yet still remain undetected in the case that system matrices are kept constant for at least a period of $2n$ time steps. It also prescribes a means to perform perfect identification.\\
\textbf{Design Recommendations}
\begin{enumerate}
\item  For all pairs $i \neq j \in \{1, \cdots l\}$, $\Lambda^i \cap \Lambda^j = \emptyset$.
\item  The system matrices $(A_k,C_k)$ are periodically changed after every $N \ge 2n$ time steps. 
\item  Let $\{l_k\}$ be a sequence where $l_k \in \{1, \cdots, l\}$. Let $q_k$ denote the indices of a subsequence. $\mbox{Pr}((A_{q_k},C_{q_k}) = (A(l_k),C(l_k)),~ \forall k) = 0$.
\item  The pair $(A(i),C(i))$ is observable all $i \in \{1, \cdots, l\}$.
\item  For all $i \in \{1, \cdots l\}$, $0 \notin \Lambda^i$.  
\end{enumerate}

\begin{corollary}
Assume a defender follows the design recommendations. Suppose sensor $s$ is attacked and there is no $t^*$ such that $D^s d_k^a = 0$ for all $k \ge t^*$. Then, the sensor attack will be unambiguously identifiable with probability 1. \label{cor:design}
\end{corollary} 
\begin{proof}
Since the attack is persistently nonzero, the adversary must guess a correct subsequence of system matrices infinitely many times due to recommendations 1 and 2. From recommendation 3, this occurs with probability 0.
\end{proof}
As a result, an attacker who persistently biases a sensor will be perfectly identified. Note that recommendation 3 can be achieved with an IID assumption or an aperiodic and irreducible Markov chain. The last 2 recommendations are justified next when we consider stochastic systems.

\section{False Data Injection Detection } \label{Stochastic}
In this section, we examine the effectiveness of the moving target defense for  detection in the case of a stochastic system. Here, we assume that
\begin{align}
x_{k+1} = A_k x_k + w_k,~~ y_k = C_k x_k + D d_k^a + v_k. \label{eq:timevaryingstoc}
\end{align}
$w_k$ and $v_k$ are IID Gaussian process and sensor noise where $w_k \sim \mathcal{N}({0},Q)$ and $v_k \sim \mathcal{N}({0},R)$. For notational simplicity we assume that the covariances $Q \ge 0$ and $R > 0$ are constant. However, we can obtain the ensuing results even in the case that $Q$ and $R$ are dependent on $A_k$ and $C_k$. 

The adversary's and defender's information and strategy is unchanged except we assume the defender has knowledge of the distribution of the initial state. Specifically, $f(x_0|\mathcal{I}_{-1}) = \mathcal{N}(\hat{x}_0^-,P_{0|-1})$. Moreover, both the defender and attacker are aware of the noise statistics. We first would like to show that a moving target defense leveraging the design recommendations listed above can almost surely detect harmful false data injection attacks. To characterize detection performance, we consider the additive bias the adversary injects on the normalized residues $\Delta z_k$ due to his sensor attacks. The residues, $z_k$, are the normalized difference between the observed measurements and their expected values. The bias on the residues is given by 
\begin{align}
\Delta e_k &= (A_{k-1}-K_kC_kA_{k-1}) \Delta e_{k-1} - K_k D d_k^a,\nonumber \\
\Delta z_k &= \mathcal{P}_k^{-\frac{1}{2}}\left( C_k A_{k-1} \Delta e_{k-1} + D d_k^a \right), ~\Delta e_0 = 0, \nonumber \\
 \mathcal{P}_k &= (C_k P_{k|k-1} C_k^T + R),~ K_k  = P_{k|k-1} C_k^T \mathcal{P}_k^{-1}, \nonumber \\
P_{k+1|k} &= A_k(P_{k|k-1} - K_kC_kP_{k|k-1})A_k^T + Q, 
\end{align}
where $\Delta e_k$ is the bias injected on the a posteriori state estimation error obtained by an optimal Kalman filter, and $P_{k|k-1}$ is the a priori error covariance. In \cite{weerakkody2016informationflow}, the authors show through their study of information flows in control systems that the residue bias in additive attacks is related to the optimal decay rate for the probability of false alarm. 
\begin{theorem}
Let $0 < \delta < 1$. Define $\alpha_k$ and $\beta_k$ as the probability of false alarm and detection respectively. Suppose $\limsup_{k \rightarrow \infty} \frac{1}{2(T+1)} \sum_{k=0}^{T} \Delta z_k^T \Delta z_k \ge \epsilon$. Then there exists a detector such that \\
$\beta_k \ge 1 - \delta,\forall k, ~~ \limsup_{k \rightarrow \infty} -\frac{1}{k+1} \log(\alpha_k) \ge \epsilon$.\\
Alternatively suppose $y_{0:k}$ is ergodic under attack and normal operation and that $\lim_{k \rightarrow \infty} \frac{1}{2(T+1)} \sum_{k=0}^{T} \Delta z_k^T \Delta z_k \le \epsilon$. Then for all detectors \\
$\beta_k \ge 1 - \delta,\forall k \implies \limsup_{k \rightarrow \infty} -\frac{1}{k+1} \log(\alpha_k) \le \epsilon$
\end{theorem}
We now show that an admissible adversary is restricted in the bias he can inject on the state estimation error without significantly biasing the residues and incurring detection. In particular, we have the following result.
\begin{theorem}
Suppose a defender uses a moving target defense leveraging the design recommendations listed above. Then for all admissible attack strategies $\limsup_{k \rightarrow \infty} \| \Delta e_k \| = \infty \implies  \limsup_{k \rightarrow \infty} \| \Delta z_k \| = \infty$. \label{theorem:fdi}
\end{theorem}
\begin{proof}
Assume to the contrary that the residues are bounded $\| \Delta z_k \|  \le M$. Define the indices of a peak subsequence as follows. $i_0 = 0$, $i_k = \min {\kappa} \mbox{ such that } \kappa > i_{k-1}, \|\Delta e_{\kappa} \| > \|\Delta e_{t} \| ~ \forall t \le \kappa$. Such a sequence exists since the estimation bias is unbounded. Also define the indices $j_k$ such that $j_k = \min{\kappa} \mbox{ such that } j_k \ge i_k, ~ j_k \mod N = N-1$
Observe that 
\begin{equation}
\Delta e_k = A_{k-1} \Delta e_{k-1} - K_k \mathcal{P}_k^{\frac{1}{2}}  \Delta z_k.
\end{equation}
As a result, we have
\begin{equation*}
A_{j_{k}} \Delta e_{j_k}  = A_{i_{k}}^{j_k - i_k + 1} \Delta e_{i_k} - \sum_{t = i_k+1}^{j_k}  A_{i_{k}}^{j_k+1 -t} K_t \mathcal{P}_{t}^{\frac{1}{2}}  \Delta z_t.
\end{equation*}
Define $a_{m} > 0$ and $a_{M} > 0$ as 
\begin{equation*}
a_m \triangleq \min_{\underset{q \in \{0,\cdots,N\}}{j \in \{1,\cdots,l\}}} \sigma_{min}(A(j)^q),~~ a_{M} \triangleq \max_{\underset{q \in \{0,\cdots,N-1\}}{j \in \{1,\cdots,l\}}} \|A(j)^q\|.
\end{equation*}
where $\sigma_{min}(\cdot)$ denotes the smallest singular value. Moreover let $p_{M}$ and $c_{M}$ be given by
\begin{equation*}
p_{M} = \sup_{k} \|P_{k|k-1}\|,~~~ c_{M} = \max_{j \in \{1,\cdots,l\}} \|C(j)\|.
\end{equation*}
Observe that $a_{m}$ is nonzero since each $A(i)$ is invertible from recommendation 5. $a_{M}$ and $c_{M}$ are bounded above since we are taking the maximum over a finite set of bounded elements. Moreover, $p_{M}$ is bounded above since the error covariance is bounded above. A complete argument is omitted due to space considerations. However, since all pairs $(A(i),C(i)) \in \Gamma$ are observable from recommendation 4 it can be shown that $x_{Nk+n}, k \in \mathbb{N}$ is a linear combination of $y_{Nk:Nk+n-1}$ and $2n$ random variables, where the linear combination is dependent only on $(A(Nk),C(Nk))$. Thus, the covariance of $x_{Nk+n}$ given   $y_{0:Nk+n-1}$ is bounded. It can be shown that the covariance of $x_{Nk+n+j}, j \in \{1,\cdots,N-1\}$ is bounded given $y_{0:Nk+n+j-1}$ simply by computing predictive covariances given $y_{0:Nk+n-1}$. As a result, we have
\begin{equation*}
\| A_{j_{k}} \Delta e_{j_k} \| \ge a_{m} \|\Delta e_{i_k}\| - (N-1) a_{M} p_{M} c_{M}\frac{M}{\sqrt{\lambda_{min}(R)}}.
\end{equation*}
where $\lambda_{min}(R)$ is the smallest eigenvalue of $R$. Therefore, since $\|\Delta e_{i_k}\| \rightarrow \infty$, we have that $\| A_{j_{k}} \Delta e_{j_k} \| \rightarrow \infty.$ Now, with some abuse of notation let $Dd_{t_1:t_2}^a = \begin{bmatrix} (Dd_{t_1}^a)^T \cdots (Dd_{t_2}^a)^T \end{bmatrix}^T$. Suppose $(A_{j_{k+1}}, C_{j_{k+1}}) = (A(q_1), C(q_1))$. Then,
\begin{equation*}
Dd_{j_k+1:j_k+N} = \mathcal{O}_{N,q_1}^{S} A_{j_{k}} \Delta e_{j_k} + F_{j_k+1}(q_1) \Delta z_{j_k+1:j_k+N}^{(q_1)}.
\end{equation*}
Through a similar analysis as done above, it can shown that $\|F_{j_k+1}(q_1)\|$ is bounded above. Alternatively, if  $q_1 \neq q_2$, is chosen then
\begin{equation*}
Dd_{j_k+1:j_k+N} = \mathcal{O}_{N,q_2}^{S} A_{j_{k}} \Delta e_{j_k} + F_{j_k+1}(q_2) \Delta z_{j_k+1:j_k+N}^{(q_2)}.
\end{equation*}
Thus, to insert valid inputs for system states $q_1$ and $q_2$, we require
\begin{align*}
(\mathcal{O}_{N,q_1}^{S} - \mathcal{O}_{N,q_2}^{S}) A_{j_{k}} \Delta e_{j_k} &=  F_{j_k+1}(q_2) \Delta z_{j_k+1:j_k+N}^{(q_2)} \\ &-  F_{j_k+1}(q_1) \Delta z_{j_k+1:j_k+N}^{(q_1)}.
\end{align*}
The right hand side is bounded since $\Delta z_{j_k+1:j_k+N}$ and $F_{j_k+1}(q_1)$ are bounded.  Due to design recommendations 1 and 2 and Theorem \ref{theorem:det}, there is no solution to $0 \neq \mathcal{O}_{N,q_1}^{S}v_1 =  \mathcal{O}_{N,q_1}^{S}v_2$.  As a result, since $A_{i_{k}} \Delta e_{j_k}$ is unbounded and each $(A_k,C_k)$ pair is observable, $(\mathcal{O}_{N,q_1}^{S} - \mathcal{O}_{N,q_2}^{S})$ has no nontrivial null space and the left hand side is unbounded. Thus, there is no way for the attacker to guess incorrectly and insert bounded residues. From, recommendation 3, there is a nonzero probability the attacker guesses $(A_{j_k+1},C_{j_k+1})$ incorrectly and the result holds.
\end{proof}
\begin{remark}
From the proof of Theorem \ref{theorem:fdi}, it can be seen that the error covariance associated with the moving target remains bounded since each system state is observable. Thus, some estimation performance is guaranteed. Moreover, since the system matrices are kept constant for a period of time, the worst case error covariance, will practically be close to the worst case LTI steady state covariance associated with the pairs $(A(j),C(j)) \in \Gamma$ and $Q$ and $R$. \label{boundcov}
\end{remark}
\begin{remark}
Design recommendation 1 can be relaxed in the stochastic case for purposes of detecting false data injection attacks. In particular for all non-equal pairs $i,j \in \{1,\cdots,l\}$ we only require $0 \neq \mathcal{O}_{N,i}^{S}v \neq  \mathcal{O}_{N,j}^{S}v$ for all $v$ instead of $0 \neq \mathcal{O}_{N,i}^{S}v_i \neq  \mathcal{O}_{N,j}^{S}v_j$ for all $v_i, v_j$. Here, a big difference is that in the stochastic case we give the defender some knowledge of the distribution of the initial state.
\end{remark}

\section{Robust Estimation and Identification} \label{Stochastic2}
While the moving target approach guarantees we can detect unbounded false data injection attacks, we wish to also identify specific malicious sensors as in the deterministic case. In the remainder of this section, we construct a robust estimator. We will fuse state estimates generated by individual sensors since previous results \cite{sun2004multi,gan2001comparison} suggest such an estimator has better fault tolerance. This is desirable in our work since we are attempting to force a normally stealthy adversary to generate faults. We will show that an attacker can destabilize this estimator only if the culprit sensors can be identified. In particular, we will show that the estimation error will become unbounded only if the bias on a sensor residue is also unbounded.

To begin, we assume that for each sensor $s$,
\begin{equation}
NS(\mathcal{O}_{n,1}^{s}) = NS(\mathcal{O}_{n,2}^{s}) = \cdots = NS(\mathcal{O}_{n,l}^{s}),
\end{equation}
where $NS(A)$ denotes the null space of $A$. Such a condition is realistic since it implies that changing the system dynamics does not affect what portion of the state the sensor itself can observe. As a result, using the Kalman decomposition, for each sensor $s$, there exists a state transformation $T_s = [T_s^{uo} ~ T_s^o]$  such that $ [T_s^{uo} ~ T_s^o] \begin{bmatrix} \zeta_{k,s}^{uo~T} \zeta_{k,s}^{T} \end{bmatrix}^T = x_k$ and $ [T_s^{uo} ~ T_s^o] \begin{bmatrix} \omega_{k,s}^{1~T} \omega_{k,s}^{T} \end{bmatrix}^T = w_k$. Here,  the columns of $T_s^{uo}$ are a basis for $NS(\mathcal{O}_{n,1}^{s})$. 

Moreover, using the same transform $T_s$, there exists a  $\Gamma^s = \{(C_s(1),A_s(1)),\cdots (C_s(l),A_s(l))\}$ corresponding to $\Gamma$ such that
\begin{align}
\zeta_{k+1,s} = A_{k,s} \zeta_{k,s} + \omega_{k,s},~~~y_k^s = C_{k,s} \zeta_{k,s} + v_k^s,
\end{align}
where each pair $(A_{k,s},C_{k,s})$ is observable and belongs to $\Gamma^s$. By performing a change of variables on $\hat{x}_0^-$, a Kalman filter with bounded covariance (see remark \ref{boundcov}) can be constructed to estimate $\zeta_{k,i}$ given $y_{0:k}^i$. Specifically, define
\begin{equation}
\begin{bmatrix} - \\ \hat{\zeta}_{0,s}^{-}    \end{bmatrix}  \triangleq T_s^{-1} \hat{x}_0^- , ~~ \begin{bmatrix} - & - \\ - & Q_{s_1,s_2} \end{bmatrix} \triangleq T_{s_1}^{-1} Q T_{s_2}^{-1~T}
\end{equation}
\begin{align}
 \begin{bmatrix} - & - \\ - & P_{0|-1}^{s_1,s_2} \end{bmatrix} \triangleq  T_{s_1}^{-1} P_{0|-1} T_{s_2}^{-1~T}
\end{align}
Then we have
\begin{align}
\hat{\zeta}_{k,s} &= (I - K_{k,s} C_{k,s}) \hat{\zeta}_{k,s}^- + K_{k,s} y_k^s, \label{eq:kalmaneqs}\\
K_{k,s} &= P_{k|k-1}^{s,s} C_{k,s}^T (C_{k,s}  P_{k|k-1}^{s,s} C_{k,s}^T + R_{ss})^{-1} \nonumber \\
P_{k}^{s_1,s_2} &= (I - K_{k,s_1} C_{k,s_1}) P_{k|k-1}^{s_1,s_2} (I - K_{k,s_2} C_{k,s_2})^T 
\nonumber \\  &+ K_{k,s_1}R_{s_1s_2}K_{k,s_2}^T, \nonumber \\
P_{k+1|k}^{s_1,s_2} &= A_{k,s_1} P_{k}^{s_1,s_2} A_{k,s_2}^T + Q_{s_1,s_2},~~ \hat{\zeta}_{k+1,s}^- = A_{k,s} \hat{\zeta}_{k,s}, \nonumber \\
z_{k,s} &= (C_{k,s}  P_{k|k-1}^{s,s} C_{k,s}^T + R_{ss})^{-\frac{1}{2}}(y_k^s - C_{k,s}\hat{\zeta}_{k,s}^-), \nonumber 
\end{align}
Here $\hat{\zeta}_{k,s} = \mathbb{E}[\zeta_k | y_{0:k}^s]$, $\hat{\zeta}_{k,s}^- = \mathbb{E}[\zeta_k | y_{0:k-1}^s]$ are optimal estimates of the reduced state for sensor s. Moreover, $P_{k}^{s_1,s_2}, P_{k+1|k}^{s_1,s_2}$ are a posteriori and a priori error covariances associated with  $\hat{\zeta}_{k,s_1}, \hat{\zeta}_{k,s_2}$ and $\hat{\zeta}_{k,s_1}^-, \hat{\zeta}_{k,s_2}^-$ given by
\begin{align*}
P_{k}^{s_1,s_2} &= \mathbb{E}[e_{k,s_1} e_{k,s_2}^T | y_{0:k}^{s_1}, y_{0:k}^{s_2}], \\
P_{k|k-1}^{s_1,s_2} &= \mathbb{E}[e_{k,s_1}^- e_{k,s_2}^{- T} | y_{0:k-1}^{s_1}, y_{0:k-1}^{s_2}],
\end{align*}
where $e_{k,s} = {\zeta}_{k,s}- \hat{\zeta}_{k,s}$ and $e_{k,s}^- = {\zeta}_{k,s} - \hat{\zeta}_{k,s}^-$. We would like to use the individual state estimates $\hat{\zeta}_{k,s}$ associated with each sensor $s$ to obtain an overall state estimate of $x_k$. To do this, first define $x_{k,s}^o$ as
\begin{equation}
x_{k,s}^o = T_s^{o}\hat{\zeta}_{k,s} + \eta_{k,s} \label{eq:addnoise}
\end{equation}
where $\eta_{k,s}$ is an IID sequence of Gaussian random variables with $\eta_{k,s} \sim \mathcal{N}(\mathbf{0},\epsilon I)$ for some small $\epsilon > 0$. Moreover $\{\eta_{k,s_1}\}$ and $\{\eta_{k,s_2}\}$ are independent sequences. $\eta_{k,s}$ is a mathematical artifact introduced so the subsequent estimator has a simplified closed form and can be easily removed or mitigated by letting $\epsilon$ tend to 0. Now, we observe that
\begin{equation}
x_k = T_s^{uo} \zeta_{k,s}^{uo} + x_{k,s}^o + T_s^{o} {e}_{k,s} - \eta_{k,s}.
\end{equation}
From here we obtain
\begin{align}
\mathbf{\hat{y}_k} &= W \mathbf{x_k} + \mathbf{\eta_k}, \label{eq:fullobserv} \\
\mathbf{\hat{y}_k} &= \begin{bmatrix}  x_{k,1}^o \\ x_{k,2}^o \\ \cdots \\ x_{k,m}^o \end{bmatrix}, ~ \mathbf{x_k} = \begin{bmatrix} \zeta_{k,1}^{uo} \\ \zeta_{k,2}^{uo} \\ \vdots \\ \zeta_{k,m}^{uo} \\ x_k \end{bmatrix},~\mathbf{\eta_k} = 
\begin{bmatrix} -T_1^{o} e_{k,1} + \eta_{k,1} \\ -T_2^{o} e_{k,2} + \eta_{k,2} \\ \vdots \\ -T_m^{o} e_{k,1} + \eta_{k,m} \end{bmatrix}, \nonumber \\
W & = \begin{bmatrix} -T_1^{uo} & \mathbf{0} & \cdots & \mathbf{0} & I \\ \mathbf{0} & -T_2^{uo} & \cdots & \mathbf{0} & I \\ 
\vdots & \vdots & \ddots & \vdots & \vdots \\ \mathbf{0} & \mathbf{0} & \cdots & -T_m^{uo} & I \end{bmatrix}
\begin{bmatrix} \zeta_{k,1}^{uo} \\ \zeta_{k,2}^{uo} \\ \vdots \\ \zeta_{k,m}^{uo} \\ x_k \end{bmatrix}. \nonumber
\end{align}
It can be seen that $\mathbf{\eta_k}$ is normally distributed so that $\mathbf{\eta_k} \sim \mathcal{N}({0},\mathcal{Q})$, where $\mathcal{Q} > 0$ consists of $m \times m$ blocks where the $(i,j)$ block is given by $(T_i^o P_{k}^{i,j} T_j^{oT}+ \delta_{ij}\epsilon I)$. Here, $\delta_{ij}$ is the Kronecker delta. The minimum variance unbiased estimate (MVUB) \cite{scharf} of $\mathbf{x}_k$ given $\mathbf{\hat{y}_k}$ is given by 
\begin{equation}
\mathbf{\hat{x}}_{k} = (W^T \mathcal{Q}^{-1} W)^{-1} W^T \mathcal{Q}^{-1} \mathbf{\hat{y}_k} \label{eq:optest}
\end{equation}
The last $n$ entries of $\mathbf{\hat{x}_{k}}$, denoted as $\hat{x}_k^*$, constitute a (MVUB) estimate of $x_k$ given the set of sensor estimates $\mathbf{\hat{{y}}_k}$. The covariance of this estimate is given by
\begin{equation}
\mbox{cov}(\mathbf{x}_k - \mathbf{\hat{x}}_{k}) = (W^T \mathcal{Q}^{-1} W)^{-1}.
\end{equation}
The proposed estimator is well defined since $NS(W) = 0$. If, $W$ had a nontrivial null space, this would imply there exists $x \neq 0$ and $\zeta_1, \cdots, \zeta_m$ such that 
\begin{equation}
T_{i}^{uo} \zeta_i = x, ~~~ \forall~i \in \{1,\cdots,m\}.
\end{equation}
This would imply that $\cap_{i=1}^m NS(\mathcal{O}_{n,1}^i) \neq 0$, which contradicts the observability of each pair $(A_k,C_k)$. We next show that the proposed estimator of $x_k$ has bounded covariance.
\begin{theorem}
Consider the estimator of $x_k$ defined by \eqref{eq:kalmaneqs},\eqref{eq:addnoise},\eqref{eq:fullobserv},\eqref{eq:optest}. The estimator has bounded covariance. \label{theorem:boundedcov}
\end{theorem}
\begin{proof}
We first prove that $\mathcal{Q} > 0$ is bounded above. The $i$th diagonal block of $\mathcal{Q}$ has covariance $(T_i^o P_{k}^{i,i} T_i^{oT}+ \epsilon I).$ Using the same argument as in the proof of Theorem \ref{theorem:fdi}, we see that $P_{k}^{i,i}$ is bounded. Consequently $(T_i^o P_{k}^{i,i} T_i^{oT}+ \epsilon I)$ and $\mathcal{Q}$ are bounded.

Next consider $\mathbf{\hat{x}_k^{uw}} = (W^TW)^{-1}W^T \mathbf{\hat{y}_k}$. Since $\mathbf{x}_k - \mathbf{\hat{x}_k^{uw}} = -(W^TW)^{-1}W^T \mathbf{\eta_k}$, $\mathbf{\hat{x}_k^{uw}}$ is an unbiased estimator of $\mathbf{x}_k$ with covariance $(W^TW)^{-1}W^T \mathcal{Q} W (W^TW)^{-1}$. Since $W$ is fixed and $\mathcal{Q} > 0$ is bounded above, $(W^TW)^{-1}W^T \mathcal{Q} W (W^TW)^{-1}$ is also bounded above.

Finally, since the proposed estimator is MVUB, we see 
\begin{equation*}
\mbox{tr}\left((W^T \mathcal{Q}^{-1} W)^{-1}\right) \le \mbox{tr}\left((W^TW)^{-1}W^T \mathcal{Q} W (W^TW)^{-1}\right).
\end{equation*}
Thus, $\mbox{cov}(\mathbf{x_k} - \mathbf{\hat{x}}_{k})$ and the covariance of $x_k$ defined by the last $n \times  n$ block of $\mbox{cov}(\mathbf{x_k} - \mathbf{\hat{x}_{k}})$ are bounded.
\end{proof}

To close this section, we demonstrate that the proposed estimator is sensitive to biases in individual residues $\Delta z_{k,s}$, specifically showing that an infinite bias introduced into the estimator implies that the residues are also infinite. Define $\mathbf{e_k} \triangleq \mathbf{x_k} - \mathbf{\hat{x}_{k}}$ and $\Delta \mathbf{e}_k$ as the bias inserted on $\mathbf{e_k}$ due to the adversary's inputs. Moreover, let $ e_k^* = x_k - \hat{x}_k^*$ and let $\Delta e_k^*$ and $\Delta e_{k,i}$ be the bias inserted on $e_k^*$ and $e_{k,i}$ respectively due to the adversary's inputs. We have the following result.
\begin{theorem}
Consider the estimator of $x_k$ defined by \eqref{eq:kalmaneqs},\eqref{eq:addnoise},\eqref{eq:fullobserv},\eqref{eq:optest}. Then $\limsup_{k \rightarrow \infty} \| \Delta e_k^* \| = \infty \implies \limsup_{k \rightarrow \infty} \| \Delta z_k^i \| = \infty$ for some $i \in \{1, \cdots, m\}$.
\end{theorem} 
\begin{proof}
First, we observe that \\ $\mathbf{e}_k = -(W^T \mathcal{Q}^{-1} W)^{-1} W^T \mathcal{Q}^{-1} \mathbf{\eta_k}$. As a result,
\begin{equation}
\Delta \mathbf{e}_k = (W^T \mathcal{Q}^{-1} W)^{-1} W^T \mathcal{Q}^{-1} T_{diag} \Delta e_{k,S} .\label{eq:estbias}
\end{equation}
where 
\begin{equation*}
T_{diag} = \begin{bmatrix} T_{1}^{o} & \mathbf{0} & \cdots & \mathbf{0} \\ \mathbf{0} & T_{2}^{o} & \cdots & \mathbf{0} 
\\ \vdots & \vdots & \ddots & \vdots \\ \mathbf{0} & \mathbf{0} & \cdots & T_{m}^o \end{bmatrix},~~
\Delta e_{k,S} = \begin{bmatrix} \Delta e_{k,1} \\ \Delta e_{k,2} \\ \vdots \\ \Delta e_{k,m} \end{bmatrix}.
\end{equation*}
Next, we will show that $K^* = (W^T \mathcal{Q}^{-1} W)^{-1} W^T \mathcal{Q}^{-1} T_{diag}$ has bounded norm. In particular observe that
\begin{align*}
\| K^* \| &=  \|(W^T \frac{\mathcal{Q}^{-1}}{\|\mathcal{Q}^{-1}\|} W)^{-1} W^T \frac{\mathcal{Q}^{-1}}{\|\mathcal{Q}^{-1}\|} T_{diag} \|  \\
&\le \|(W^T \frac{\mathcal{Q}^{-1}}{\|\mathcal{Q}^{-1}\|} W)^{-1}\| \|W^T \frac{\mathcal{Q}^{-1}}{\|\mathcal{Q}^{-1}\|} T_{diag} \|.
\end{align*}
Clearly, $\|W^T \frac{\mathcal{Q}^{-1}}{\|\mathcal{Q}^{-1}\|} T_{diag} \|$ has bounded norm. Moreover, using a similar argument as in the proof of Theorem \ref{theorem:boundedcov}, $ \|(W^T \frac{\mathcal{Q}^{-1}}{\|\mathcal{Q}^{-1}\|} W)^{-1}\| $ has bounded norm. Thus, $\|K^*\|$ is bounded. Consequently, from \eqref{eq:estbias}, $\limsup_{k \rightarrow \infty} \| \Delta e_k^* \| = \infty \implies \limsup_{k \rightarrow \infty} \| \Delta e_{k,i} \| = \infty$ for some $i \in \{1,\cdots,m\}$. However, from Theorem \ref{theorem:fdi}, this implies $\limsup_{k \rightarrow \infty} \| \Delta z_{k}^i \| = \infty$ and the result holds.
\end{proof}
While the proposed estimator does not guarantee each malicious sensor will be identified, it does guarantee that the defender will be able to identify and remove sensors whose attacks cause unbounded bias in the estimation error simply by analyzing each sensor's measurements individually. This is due to the fact that the bias on residues of such sensors will grow unbounded, which can be easily detected by some $\chi^2$ detector. As a result, for each individual sensor $s$, we propose the following detector at time $k$, which can be used to identify malicious behavior,
\begin{equation}
\sum_{j = k-T^*+1}^k z_{j,s}^2 \overset{{H_1^s}}{\underset{H_0^s}{ \gtrless}} \gamma \label{chi}.
\end{equation}
In this scenario, $H_1^s$ is the hypothesis that sensor $s$ is malfunctioning and $H_0^s$ is the hypothesis that sensor $s$ is working normally. In practice a sensor $s$ who repeatedly fails detection can be removed from consideration when obtaining a state estimate and the proposed fusion based estimation scheme can be adjusted accordingly.

\section{Numerical Example} \label{Simulation}
We consider a numerical example where $l = 7$ and $A(j)$ and $C(j)$ are given by
\begin{equation*}
A(k) = \begin{bmatrix} A_{11}(j) & A_{12}(j) & 0 & 0 & 0 \\ 0 & A_{22}(j) & 0 & A_{24}(j) & 0 \\ 0 & 0 & A_{33}(j) & 0 & A_{35}(j) \\ 0 & 0 & 0 & A_{44}(j) & A_{45}(j) \\ 0 & 0 & 0 & 0 & A_{55}(j) \end{bmatrix},
\end{equation*}
\begin{equation*}
C(j) = \begin{bmatrix} C_1(j) \\ C_2(j) \end{bmatrix},
\end{equation*}
\begin{equation*}
  C_i(j) = \begin{bmatrix} C_{1,i}(j) & 0 & 0 & 0 & 0 \\ 0 & C_{2,i}(j) & 0 & 0 & 0 \\ 0 & 0 & C_{3,i}(j) & 0 & 0 \\ 0 & 0 & 0 & C_{4,i}(j) & 0 \\ 0 & 0 & 0 & 0 & C_{5,i}(j) \end{bmatrix}.
\end{equation*}
where $A_{ij}(j) \in \mathbb{R}^{3 \times 3}$ and $C_{i,j}(j) \in \mathbb{R}^{1 \times 3}$ are scaled uniformly random matrices with $A_{ii}(j)$ unstable. Moreover $Q$ and $R$ are appropriately sized matrices generated by multiplying a uniform random matrix by its transpose. The system matrices are changed independently and randomly every $2n$ time steps where $n = 15$ and each $(A(j),C(j))$ pair has equal likelihood.

We assume that the adversary biases the last $5$ sensors (measured by $C_2(j)$) by performing the attack formulated in Theorem \ref{MTImage}. Here, the attacker guesses the system matrices randomly every $2n$ time steps and $x_0^*$ is chosen identically for each sensor. A $\chi^2$ detector \eqref{chi} with window $5$ and false alarm probability $\alpha_k^i = 6.9 \times 10^{-8}$ is implemented for each sensor based on their local Kalman filters. A centralized $\chi^2$ detector with window $3$ derived from the optimal centralized Kalman filter performs detection with false alarm probability $\alpha_k = 4.2 \times 10^{-4}$.

In Fig. \ref{MSE}, we compared the optimal Kalman filter and the proposed fusion based estimator under normal operation by plotting their expected mean squared error as a function of time. The estimators achieve similar performance, with the average mean squared of the optimal Kalman filter at 22.9 and the average mean squared error of the proposed estimator at 23.4.
\begin{figure}
\begin{center}
\includegraphics[scale=0.43]{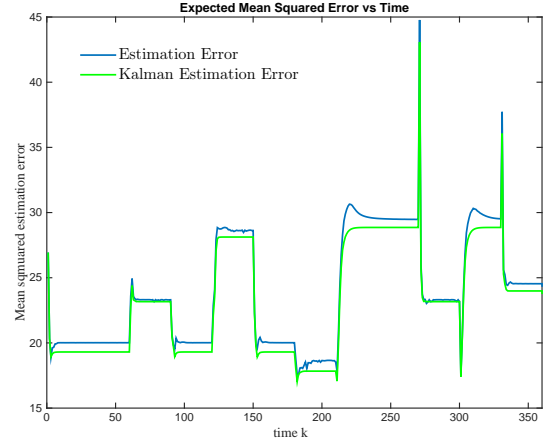}
\end{center}
\caption{Expected mean squared error vs time under normal operation for proposed fusion based estimator and the centralized Kalman filter.}
\label{MSE}
\end{figure}

In Fig. \ref{unstableerror} and \ref{residues}, we consider the system with the moving target under attack. However, we assume the attacker is aware of the exact sequence of time varying matrices. As such the attacker is able to destabilize the estimation error in Fig. \ref{unstableerror} while the sensor residues appear normal in Fig. \ref{residues}.

\begin{figure}
\begin{center}
\includegraphics[scale=0.43]{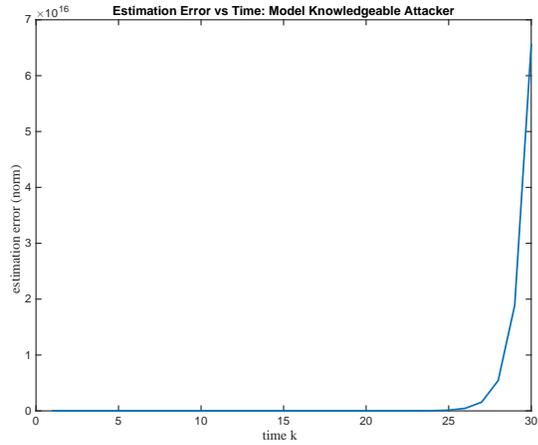}
\end{center}
\caption{Estimation error vs time under attack when the adversary knows the system dynamics. }
\label{unstableerror}
\end{figure}

\begin{figure}
\begin{center}
\includegraphics[scale=0.43]{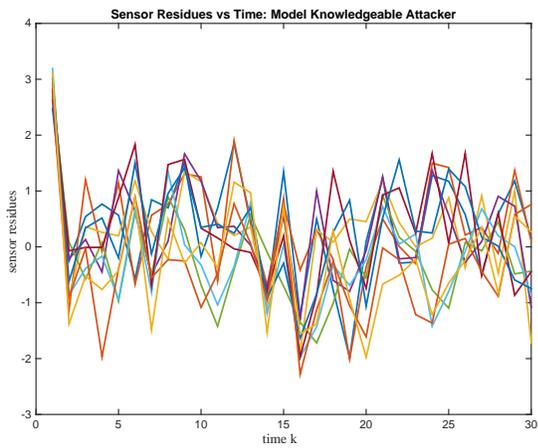}
\end{center}
\caption{Residue vs time under attack when the adversary knows the system dynamics.  }
\label{residues}
\end{figure}

Finally, in Fig. \ref{identfig}, we plot the norm of the estimation error for both the proposed estimator and optimal Kalman filter as a function of time when the attacker is forced to randomly guess the system model. Here the attacker is detected in $2$ time steps and perfectly identified in $10$ time steps. When a sensor is identified, it is removed from consideration when performing fusion or optimal Kalman filtering. It can be seen that while under attack, the proposed fusion based estimator is better able to recover from the adversary's actions.

\begin{figure}
\begin{center}
\includegraphics[scale=0.43]{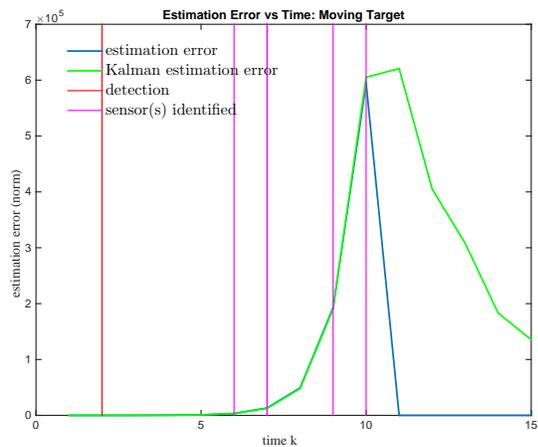}
\end{center}
\caption{Estimation error vs time under attack when the adversary does not know the true dynamics. All sensor attacks are identified. The proposed fusion based estimator and the centralized Kalman filter are illustrated.}
\label{identfig}
\end{figure}

\section{Conclusion} \label{Conclusion}
In this paper, we consider a moving target approach for identifying sensor attacks in control systems. We first considered the deterministic case and examined design considerations which ensure perfect identification. We then considered the stochastic scenario and showed that the proposed solution can effectively prevent infinite false data injection attacks. Finally, we constructed a robust estimator, which allows us to identify sensors which cause unbounded estimation error. Future work consists in examining and evaluating a larger class of estimators in the stochastic case including the traditional Kalman filter. Furthermore, we wish to investigate the amount of bias the attacker can introduce into the estimation error while avoiding identification with tests on real systems. Finally, we wish to consider more powerful adaptive adversaries who observe a subset of sensor measurements and perform system identification.

\bibliographystyle{IEEEtran}
\bibliography{identification_cps2}

\section{Appendix: Proof of Theorem 12}
We first proof sufficiency. Suppose $\exists \lambda \in \Lambda^1 \cap \Lambda^2$ and $\alpha_1 \in \mathbb{C}^{\sum_i{r_i(\lambda,1)}},\alpha_2 \in \mathbb{C}^{\sum_i{r_i(\lambda,2)}}$ such that 
\begin{equation}
\mathcal{V}_s^{\lambda,1}\alpha_1 =  \mathcal{V}_s^{\lambda,2}\alpha_2 \neq 0. \label{eq:sameim}
\end{equation}
Let $\bar{V}^{\lambda,j}$ be given by
\begin{equation*}
\begin{bmatrix} v_{1,1}^{\lambda,j} & \cdots & v_{r_1,1}^{\lambda,j}  & \cdots & v_{1,l_{\lambda,j}}^{\lambda,j} & \cdots & v_{r_{l_{\lambda,j}},l_{\lambda,j}}^{i,j} \end{bmatrix}
\end{equation*}
and let $x_0^a(j) = \bar{V}^{\lambda,j}\alpha_j$. Suppose 
\begin{align}
& D^sd_k^a(j) = C^s(j)A(j)^k x_0^a(j) \nonumber \\
 &= \lambda^k \mathcal{V}_{s,1}^{\lambda,j}\alpha_j  + \cdots + \frac{1}{(k-1)!} \frac{d^{k-1}}{d\lambda^{k-1}}(\lambda^k) \mathcal{V}_{s,k}^{\lambda,j} \alpha_j,~k \le r(\lambda) \nonumber \\
  &= \lambda^k \mathcal{V}_{s,1}^{\lambda,j}\alpha_j  + \cdots + \frac{1}{(r(\lambda)-1)!} \frac{d^{r(\lambda)-1}}{d\lambda^{r(\lambda)-1}}(\lambda^k) \mathcal{V}_{s,r(\lambda)}^{\lambda,j} \alpha_j, \nonumber \\
  & ~~~~~~~~~~~~~~~~~~~~~~~~~~~~~~~~~~~~~~~~~~~~~~~k > r(\lambda),  \label{eq:jordan}
\end{align}
where $\mathcal{V}_{s,k}^{\lambda,j}$ is the $k$th row of $\mathcal{V}_{s}^{\lambda,j}$. From \eqref{eq:sameim}, $D^sd_k^a(1) = D^sd_k^a(2)$. Note that if $D^sd_k^a(j)$ is not real, then an attack can be constructed by adding the conjugate so that 
\begin{align*}
D^sd_k^a(j) &= C^s(j)A(j)^k x_0^a(j)  + \overline{ C^s(j)A(j)^k x_0^a(j)} \\
&= C^s(j)A(j)^k(x_0^a(j) + \overline{x_0^a(j)})
\end{align*}
Therefore, 
\begin{equation*}
C^s(1)A(1)^k(x_0^a(1) + \overline{x_0^a(1)}) = C^s(2)A(2)^k(x_0^a(2) + \overline{x_0^a(2)}).
\end{equation*}
The attack can be scaled so it is nonzero for some time $k$ from \eqref{eq:jordan} and \eqref{eq:sameim}. Thus, the result holds.

We now prove the necessary assumption. Without loss of generality, suppose the first $z$ eigenvalues of $\Lambda^1$ and $\Lambda^2$ are the same so that $\lambda_k^1 = \lambda_k^2$ for $k \le z$. Assume the rest of the eigenvalues are different. In particular let  $\Lambda^1 = \{\lambda_1, \cdots, \lambda_{q_1}\}$ and $\Lambda^2 = \{\lambda_1, \cdots, \lambda_{z}, \lambda_{q_1+1}, \cdots, \lambda_{q_1+q_2-z}\}$. Let $r^*(\lambda,j) = \max_i r_i(\lambda,j)$, characterize the maximum block size of eigenvalue $\lambda$ for $A(j)$ and let $\tau+1 \ge 2n$.

Define $G(\lambda_i,j) \in \mathbb{C}^{\tau+1 \times r^*(\lambda_i,j)}$ as 
\begin{equation}
\begin{bmatrix} 1 & 0 & \cdots & 0 \\ \lambda_i & 1 & \cdots & 0 \\ \lambda_i^2 & 2\lambda_i & \cdots & 0 \\ \vdots & \vdots & \cdots & \vdots \\ \lambda_i^{\tau} & \tau\lambda_i^{\tau-1} & \cdots & \frac{1}{(r^*(\lambda_i,j)-1)!} \frac{d^{r^*(\lambda_i,j)-1}}{d {r^*(\lambda_i,j)-1}} (\lambda_i^{\tau}) \end{bmatrix}.
\end{equation}
where the $k$th column is obtained by taking entrywise, the corresponding $(k-1)$ derivative of the associated entry in the first column and dividing by $(k-1)!$.
Let $G_a, G_b, G_c$ be given by
\begin{align*}
G_a &= \begin{bmatrix} G(\lambda_1,1) & G(\lambda_1,2) & \cdots & G(\lambda_z,1) & G(\lambda_z,2) \end{bmatrix},\\
G_b &= \begin{bmatrix} G(\lambda_{z+1},1) & \cdots  & G(\lambda_{q_1},1) \end{bmatrix}, \\
G_c &= \begin{bmatrix} G(\lambda_{q_1+1},2) & \cdots  & G(\lambda_{q_1+q_2-z},2) \end{bmatrix}.
\end{align*}
Finally, let $G^*= \begin{bmatrix} G_a & G_b & G_c \end{bmatrix}$. Note that $G^* \in \mathbb{C}^{\tau+1 \times \kappa}$ where $\kappa \le 2n$ by construction.

Consider vectors $\eta^{i,j} \in \mathbb{C}^{\sum_{k} r_k(\lambda_i,j)}$  and define $\tilde{\mathcal{V}}_{s}^{\lambda_i,j}, \tilde{\mathcal{V}}_{s}^a, \tilde{\mathcal{V}}_{s}^b, \tilde{\mathcal{V}}_{s}^c, \tilde{\mathcal{V}}_s$ as
\begin{align*}
&\tilde{\mathcal{V}}_{s}^{\lambda_i,j} = \begin{bmatrix} \mathcal{V}_{s,1}^{\lambda_i,j} \eta^{i,j} & \cdots & \mathcal{V}_{s,r^*({\lambda_i,j})}^{\lambda_i,j} \eta^{i,j} \end{bmatrix}^T, \\
&\tilde{\mathcal{V}}_{s}^{a} = \begin{bmatrix} (\tilde{\mathcal{V}}_{s}^{\lambda_1,1})^T & (\tilde{\mathcal{V}}_{s}^{\lambda_1,2})^T & \cdots & (\tilde{\mathcal{V}}_{s}^{\lambda_z,1})^T & (\tilde{\mathcal{V}}_{s}^{\lambda_z,2})^T \end{bmatrix}^T, \\
&\tilde{\mathcal{V}}_{s}^{b} = \begin{bmatrix} (\tilde{\mathcal{V}}_{s}^{\lambda_{z+1},1})^T & \cdots & (\tilde{\mathcal{V}}_{s}^{\lambda_{q_1},1})^T \end{bmatrix}^T, \\
&\tilde{\mathcal{V}}_{s}^{c} = \begin{bmatrix} (\tilde{\mathcal{V}}_{s}^{\lambda_{q_1+1},2})^T & \cdots & (\tilde{\mathcal{V}}_{s}^{\lambda_{q_1+q_2-z},2})^T \end{bmatrix}^T, \\
&\tilde{\mathcal{V}}_s = \begin{bmatrix} (\tilde{\mathcal{V}}_{s}^{a})^T & (\tilde{\mathcal{V}}_{s}^{b})^T & (\tilde{\mathcal{V}}_{s}^{c})^T \end{bmatrix}.
\end{align*}
It can be shown (Theorem \ref{sum of nulls}) that an attack exists only if  there exists some nontrivial $\tilde{\mathcal{V}}_s$  in the null space of $G^*$. We observe by construction that $G^*$ has 
$r_{min} = \sum_{i = 1}^{z} \min_{j} r^*(\lambda_i,j)$ pairs of identical columns. Thus, $\mbox{null}(G^*) \ge r_{min}$. Let $\tilde{G^*} \in \mathbb{C}^{\tau+1 \times r_{max}}$ be obtained by deleting duplicate columns of $G^*$ where $r_{max} \le 2n \le \tau + 1$ is given by
\begin{displaymath}
\sum_{i = 1}^{z} \max_{j} r^*(\lambda_i,j)  + \sum_{i = z+1}^{q_1} r^*(\lambda_i,1)  + \sum_{i = q_1+1}^{q_2+q_1-1} r^*(\lambda_i,2).
\end{displaymath}
Let $\tilde{G^*}_{trunc}$ be a square matrix obtained by removing the last $\tau+1 - r_{max}$ rows of $\tilde{G^*}$. 

We first show that the null space  $\tilde{G^*}_{trunc}^T$ is empty. Suppose it was not. This would imply the existence of a complex nonzero polynomial $p^*(x)$ of degree $r_{max}$ with the property
\begin{align*}
p^*(\lambda_k) = 0, \cdots, \frac{d^{(\max_{j} r^*(\lambda_k,j)-1)}p^*}{dx}(\lambda_k) = 0, \
\end{align*}
for $1 \le k \le z$,
\begin{align*}
p^*(\lambda_k) = 0, \frac{dp^*}{dx}(\lambda_k) =0,\cdots, \frac{d^{r^*(\lambda_i,1)-1}p^*}{dx}(\lambda_k) = 0, \
\end{align*}
for $z+1 \le k \le q_1$, and
\begin{align*}
p^*(\lambda_k) = 0, \frac{dp^*}{dx}(\lambda_k) =0,\cdots, \frac{d^{r^*(\lambda_i,2)-1}p^*}{dx}(\lambda_k) = 0, \
\end{align*}
for $q_1+1 \le k \le q_1+q_2-z$.

But this contradicts the fundamental theorem of algebra since it would imply a polynomial of degree $\tau$ has $\tau+1$ zeros. Thus, the null space  $\tilde{G^*}_{trunc}^T$ is empty.

Therefore, by the rank nullity theorem $\tilde{G^*}_{trunc}$ is full rank and therefore $\tilde{G^*}$ is full rank. Consequently, $\mbox{rank}(G^*) \ge r_{max}$. However, $\mbox{null}{(G^*)} \ge r_{min}$ and the number of columns in $G^*$ is $r_{max}+r_{min}$. Therefore, strict equality holds and $\mbox{rank}(G^*) = r_{max}$ and $\mbox{null}{(G^*)} = r_{min}$. As a result, one excites the null space of $G^*$ only by exciting pairs of identical columns in $G^*$.

Therefore, $\tilde{\mathcal{V}}_s$ is in the null space of $G^*$  only if for $1 \le k \le z$ 
\begin{equation*}
\tilde{\mathcal{V}}_{s}^{\lambda_k,1}   = -\tilde{\mathcal{V}}_{s}^{\lambda_k,2} .
\end{equation*}
However, there exists a nonzero attack only if for some $k$ there exists a solution 
\begin{equation*}
\tilde{\mathcal{V}}_{s}^{\lambda_k,1}  = -\tilde{\mathcal{V}}_{s}^{\lambda_k,2}  \neq 0.
\end{equation*}
The result directly follows.
\end{document}